\documentclass{article}
\usepackage{theorem,graphicx,amssymb,amsmath}
\usepackage[retainorgcmds]{IEEEtrantools}
\usepackage{lineno}

\theorembodyfont{\slshape}

\newtheorem{theorem}{Theorem}
\newtheorem{lemma}[theorem]{Lemma}

\newtheorem{problem}{Problem}

\def\QED{\ensuremath{{\square}}}

\def\markatright#1{\leavevmode\unskip\nobreak\quad\hspace*{\fill}{#1}}
\newenvironment{proof}
  {\begin{trivlist}\item[\hskip\labelsep{\bf Proof.}]}
  {\markatright{\QED}\end{trivlist}}

\newcommand{\crs}{\overline{\operatorname{cr}}}

\usepackage{wasysym}
\title{Updating the Number of Crossings in Rectilinear Drawings of the Complete Graph\footnote{This work was partially supported
by Conacyt of Mexico grant 253261.}}

\author{Frank Duque\thanks{Departamento de Matem\'aticas, CINVESTAV.} \thanks{\tt{frduque@math.cinvestav.mx}} 
\and Ruy Fabila-Monroy\footnotemark[2] \thanks{\tt{ruyfabila@math.cinvestav.edu.mx}}}

\begin{document}
\date{\today}
\maketitle

\begin{abstract}
Let $S$ be a set of $n$ points in general position in the plane. 
Join every pair of points in $S$ with a straight line segment. 
Let $\overline{cr}(S)$ be number of pairs of these edges that intersect in their interior.
Suppose that this number is known.
In this paper we consider the problem of computing  $\overline{cr}(S')$,
where $S'$ comes from adding, deleting or moving  a point from $S$.
 \end{abstract}

\section{Introduction}

A set of points in the plane is in \emph{general position} if no three of its points are collinear.
Let $G$ be a graph on $n$ vertices. A \emph{rectilinear drawing} of  $G$
is a drawing of $G$  in the plane, in which its vertices are placed at points in general position, and its edges are drawn as straight line segments joining these points.

A pair of straight line segments in the plane \emph{crosses} if their intersection is a single point different from their endpoints.
The \emph{number of crossings} of a rectilinear drawing 
is the number of pairs of its edges that cross.
The \emph{rectilinear crossing number} of $G$ is the minimum number 
of crossings over all rectilinear drawings of $G$; we denote it by $\crs(G)$.

In the case when $G$ is a complete graph, 
the number of crossings in a rectilinear drawing of $G$, depends
only on the position of its vertices. Let $S$ be a set of $n$ points in general position in the plane, and 
let $\crs(S)$ be the number of crossings in a rectilinear drawing
of the complete graph $K_n$ with $S$ as its vertex set. We abuse notation
and refer to $\crs(S)$ as the number of crossings of $S$. Note that 
\[\crs(K_n)=\min \{\crs(S): S \textrm{ is a set of } n \textrm{ points in general position in the plane}\}.\]
Since this value only depends on $n$, for brevity we refer to $\crs(K_n)$ as $\crs(n)$.
The current best bounds on $\crs(n)$ are 
\[ 0.379972 \binom{n}{4} < \crs(n)  < 0.380473\binom{n}{4}+\Theta(n^3).\]
The lower bound was given by \'Abrego Fern\'andez-Merchant, 
Lea\~nos and Salazar \cite{lower};  
the upper bound  was given by Fabila-Monroy and L\'opez~\cite{crossing}. 
The computation of $\crs(n)$ is an active research problem in Combinatorial Geometry. 
For a recent survey regarding the rectilinear crossing number of $K_n$ see \cite{survey}.

 A successful approach for upper bounding $\crs(n)$ has been to
provide constructions that take a point set with few crossings and produce
a larger point set also with few crossings. The new set is then fed again to the construction to produce
an even larger set with few crossings. Applying this process iteratively, arbitrarily large 
point sets with few crossings can be found.
This approach has been refined over the years~\cite{singer,towards,oswinupper,bersil,upper}.
With the current best such construction being that of \'Abrego, Cetina,
Fern\'andez-Merchant, Lea\~nos, and Salazar \cite{upper}. 

A salient feature of this approach is that to find a new general upper bound for $\crs(n)$
it is sufficient to find a small point set with sufficiently few crossings.
In~\cite{crossing}, a simple heuristic was used for this purpose. They improved many of the best known
sets of $27,\dots, 100$ points with few crossings\footnote{These point sets were obtained from 
Oswin Aichholzer's page http://www.ist.tugraz.at/aichholzer/research/rp/triangulations/crossing/}.
In particular, they found a  set of $75$ points with $450492$ crossings. This
point set together with the construction of~\cite{upper} provide the current best upper bound
on $\crs(n)$.

The heuristic used in~\cite{crossing} is as follows.
Choose a random point $p$ of $S$, and a random point $q$ near $p$. 
Afterwards, compute $\crs (S\setminus \{p\}\cup \{q\})$. 
If this number is less or equal to $\crs(S)$ then replace $p$ with $q$ in $S$. 
The improvements obtained in~\cite{crossing} 
were done by many iterations of this procedure. Experimentally,
it seems that heuristics of these type work well in practice; recently in \cite{pseudolinear},
Balko and Kyn\v{c}l used simulated annealing to improve the best
upper bound on $K_n$ on a parameter similar to the rectilinear crossing number;
this parameter is  called the pseudolinear crossing
number.

In the heuristic of \cite{crossing}, at most one point of $S$ changes position at each step.
We have observed experimentally that
removing or adding a point from a point set with few  crossings tends to produce a point set
with few crossings.  
Thus, it is sensible to consider the following problem.
\begin{problem}\label{prob:main}
 Suppose that $S'$ is obtained from $S$ by either moving, removing or adding a point.
 Assuming that $\crs(S)$ is known: What is the time complexity of computing $\crs(S')$?
\end{problem}

Since $\crs (S)$ can be computed in $O(n^2)$ time, perhaps 
Problem~\ref{prob:main} can be solved in $o(n^2)$ time. Although we are
currently unable to solve Problem~\ref{prob:main} in subquadratic time,
we are able to prove the following amortized results on this problem.

\begin{theorem} \label{thm:move}
 Let $S$ be a set of $n$ points in the general position in the plane;
 let $C$ be a set of $\Theta(n)$ ``candidate'' points,
 such that $S \cup C$ is in general position and $C\cap S = \emptyset$.
 Let $p$ be a point in $S$.
 Then the set of values
 \[\left \{\crs(S'):S'=S\setminus \{p\} \cup \{q\}, q \in C \right \} \]
 can be computed in $O(n^2)$ time.
\end{theorem}
\begin{theorem}\label{thm:remove}
 Let $S$ be a set of $n$ points set in general position in the plane.
 Then the set of values
 \[\left \{\crs(S\setminus \{ p \}): p \in S \right \}\] can be computed in $O(n^2)$ time.
\end{theorem}
\begin{theorem}\label{thm:add}
Let $S$ be a set of $n$ points in the general position in the plane;
 let $C$ be a set of $\Theta(n)$ ``candidate'' points,
 such that $S \cup C$ is in general position and $C\cap S = \emptyset$.
 Then the set of values
 \[\{ \crs(S \cup \{q\}) :q \in C \}\] can be computed in $O(n^2)$ time.
\end{theorem}
Note that in each of these theorems the amortized time per point is linear. 

We implemented\footnote{Our implementations run in $O(n^2 \log n)$ time.} the algorithms implied by Theorems~\ref{thm:move},~\ref{thm:remove} and~\ref{thm:add}.
These implementations are being used in an ongoing project to improve the upper bound
on $\crs(n)$~\cite{project}.  This project has improved the current upper bound on $\crs(n)$; we refrain
from mentioning its value as it is not part of this paper and has not been made public yet.
Our implementations are available at \texttt{www.pydcg.org}. In Section~\ref{sec:preliminaries}, we provide
some preliminary definitions and Lemmas. In Section~\ref{sec:proofs}, we prove
Theorems~\ref{thm:move},~\ref{thm:remove} and~\ref{thm:add}.

\section{Preliminaries}\label{sec:preliminaries}

In this section we prove some lemmas that are used 
to prove Theorems~\ref{thm:move},~\ref{thm:remove} and~\ref{thm:add}.
Afterwards, we recall the concept of the $\lambda$-matrix of a point set; we also provide  a
characterization of the number of crossings of a point set  in terms of its  $\lambda$-matrix.

We frequently need to know, for every point $p \in S$, the counterclockwise order around $p$
of the points in $S \setminus \{p\}$. We have the following lemma.
\begin{lemma}\label{lem:SortPoints}
 The set of counterclockwise orders around $p$ of $S\setminus \{ p\}$, of every point $p \in S$, 
 can be computed in $O(n^2)$ time.
\end{lemma}
\begin{proof}
 Dualize $S$ to a set of $n$ lines.
 The corresponding line arrangement can be constructed in $O(n^2)$ time  with  standard algorithms. 
 The clockwise orders of $S\setminus \{ p\}$ around each $p \in S$ can then be
 extracted from this line arrangement in $O(n^2)$ time.
\end{proof}

Let $p$ be a point not in $S$. For every point $q \in S$, let  $S_p(q)$ be the set of 
 points in $S$ to the left of the directed line from $p$ to $q$.
 
\begin{lemma}\label{lem:ComputingLemma}
 Let $p$ be a point not in $S$. Suppose that each point $r \in S$ has a weight $w(r)$ assigned to it,
 and that the counterclockwise order around $p$ of the points in $S$ is known.
 Then the set of values 
 \[
  \left \{\sum_{r \in S_p(q)}{w(r)}: q\in S \right \}.
 \] can be computed in linear time.
\end{lemma}
\begin{proof}
Let $q_1$ be a point in $S$.
 Let $\ell$ be the directed line from $p$ to $q_1$. 
 Rotate $\ell$ counterclockwise around $p$, and let $(q_1,\dots,q_n)$ be the points of $S$ 
 in the order as they are encountered by $\ell$ during this rotation. This order can be computed from the counterclockwise order
 of the points in $S$ around $p$ in $O(n)$ time. Compute $\sum_{r \in S(q_1)}{w(r)}$ in $O(n)$
 time. Since $\sum_{r \in S(q_i)}{w(r)}$ can be computed from $\sum_{r \in S(q_{i-1})}{w(r)}$ in constant
 time, the result follows.
\end{proof}

\subsection*{The $\lambda$-Matrix}
Let $p$ and $q$ be a pair of points not necessarily in $S$.
Let $\lambda_S(p,q)$ be the number of points of $S$ that lie to the left
of the directed line from $p$ to $q$; in the case that $p=q$, we set $\lambda_S(p,q):=0$.
Let $p_1,p_2,\cdots,p_n$ be the points in $S$.
The $\lambda$-matrix of $S$ 
is the matrix whose $(i,j)$-entry is equal to $\lambda_S(p_i,p_j)$.
The following lemma is well known; it can be proven from
Lemma~\ref{lem:ComputingLemma} by assigning a weight equal to one
to every point in $S$.  

\begin{lemma} \label{lem:ComputingLambdaMatrix}
 The $\lambda$-matrix of a set of $n$ points in general position in the plane can be computed in $O(n^2)$ time. 
\end{lemma}
\markatright{\QED}

It is known
that the $\lambda$-matrix of $S$ determines $\crs(S)$. This was shown independently by
Lov\'asz, Wagner, Welzl, and Wesztergombi~\cite{k_edges_lovasz}, and by \'Abrego
and Fern\'andez-Merchant~\cite{k_edges_silvia}.
We now provide another characterization of $\crs(S)$ in terms of the $\lambda$-matrix of $S$.
For two any finite sets of points $P$ and $Q$,
define \[
 f_S(P,Q):=\sum_{p\in P}\sum_{q\in Q}\binom{\lambda_{S}(p,q)}{2}.
\]

%
%
%
%
%

\begin{lemma}
\label{lem:FormulaCrossingNumber}
\[\crs(S)=f_S(S,S)-\frac{n(n-1)(n-2)(n-3)}{8}.\]
\end{lemma}
\begin{proof}

Let $p,q,r$ and $s$ be four different points of $S$. 
We call the tuple $((p,q),\lbrace r,s\rbrace)$ a \emph{pattern}. 
If the points $r$ and $s$ are both to the left of the directed line from $p$ to $q$,
we say that $((p,q),\lbrace r,s\rbrace)$ is a \emph{type-A} pattern, otherwise
we call it a \emph{type-B} pattern.
We denote by $A(S)$ and $B(S)$ the number of type-A and type-B patterns in $S$, respectively.

Let $P$ be a set of four points.
If $P$ is in convex position, then $P$ determines $4$ type-A patterns and $8$ 
 type-B patterns.
If $P$ is not in convex position, then $P$ determines $3$ type-A patterns and $9$ 
type-B patterns. Let $\Square(S)$ denote the number of subsets of $S$ of four points in convex position,
and let $\triangle(S)$ denote the number of subsets of $S$ of four points not in convex position.
Thus,
\begin{IEEEeqnarray}{rCl}
A(S) & = & 4\Square(S)+3\triangle(S) \textrm{ and } \nonumber \\
 B(S) & = & 8\Square(S)+9\triangle(S). \nonumber
 \end{IEEEeqnarray}

Note that $A(S)+B(S)=n(n-1)(n-2)(n-3)/2$, $\crs(S)=\Square(S)$ and \[A(S)=\sum_{p,q \in S } \binom{\lambda_S(p,q)}{2}.\] 
Therefore,
\begin{IEEEeqnarray}{rCl}
 \crs(S)&=&A(S)-(A(S)+B(S))/4 \nonumber \\
           &=&f_S(S,S)-\frac{n(n-1)(n-2)(n-3)}{8}. \nonumber
\end{IEEEeqnarray}
\end{proof}

\section{Proofs of Theorems \ref{thm:move}, \ref{thm:remove} and \ref{thm:add}}\label{sec:proofs}

\subsection*{Proof of Theorem~\ref{thm:remove}}

 For every point $p \in S$, compute the clockwise order of $S\setminus \{ p\}$ around $p$; afterwards, 
 compute the $\lambda$-matrix of $S$. By Lemmas~\ref{lem:SortPoints} and ~\ref{lem:ComputingLambdaMatrix} this
 can be done in $O(n^2)$ time. Using the $\lambda$-matrix of $S$ compute 
 $f_S(S,S)$, $\{f_S(\{p\},S):p \in S\}$ and $\{f_S(S,\{p\}):p \in S\}$ in $O(n^2)$ time.
 
 Note that by Lemma~\ref{lem:FormulaCrossingNumber}, for every $p \in S$ we have that 
 \[\crs(S\setminus \{ p\})=f_{S\setminus \{ p\} }(S\setminus \{ p\},S\setminus \{ p\})-\frac{(n-1)(n-2)(n-3)(n-4)}{8}.\]
 Thus, it is enough to compute $\{f_{S\setminus \{ p\} }(S\setminus \{ p\},S\setminus \{ p\}): p \in S\}$ in $O(n^2)$ time,
 For every $p \in S$, let \[\nabla_p:=f_{S\setminus \{ p\}}(S\setminus \{ p\},S\setminus \{ p\})-f_S(S,S)+f_S(\{p\},S)+f_S(S,\{ p\}).\]
 To compute $\{f_{S\setminus \{ p\} }(S\setminus \{ p\},S\setminus \{ p\}): p \in S\}$ in $O(n^2)$ time, we compute
 $\{\nabla_p : p \in S\}$ in $O(n^2)$ time.
 Note that \[\nabla_p=\sum_{q \in S\setminus \{ p\}} \sum_{r \in S\setminus \{ p\}} \left( \binom{\lambda_{S\setminus \{ p\}}(q,r)}{2}-\binom{\lambda_S(q,r)}{2} \right ).\]
 
 Let $r$ be a point in $S\setminus \{ p\}$. Note that if
 $p$ is to the right of the directed line from $q$ to $r$  then \[\binom{\lambda_{S\setminus \{ p\}}(q,r)}{2}-\binom{\lambda_S(q,r)}{2}=0;\] and 
 if $p$ is to the left of the directed line from $q$ to $r$ then
 \[\binom{\lambda_{S\setminus \{ p\}}(q,r)}{2}-\binom{\lambda_S(q,r)}{2}=1-\lambda_{S}(q,r).\]  Moreover,
 $p$ is to the left of the directed line from $q$ to $r$ if and only if 
 $r$ is to the left of the directed line from $p$ to $q$.
 
 For every point $q \in S$ do the following.  
 To every point $r \in S \setminus \{q\}$ 
 assign the weight $w_q(r):=1-\lambda_{S}(q,r)$. Thus,
 \[ \sum_{r \in S\setminus \{ p\}} \left( \binom{\lambda_{S\setminus \{ p\}}(q,r)}{2}-\binom{\lambda_S(q,r)}{2} \right ) = \sum_{r \in S_p(q)} w_q(r).\]
 By Lemma~\ref{lem:ComputingLemma}, for a fixed $q \in S$, the set of values
 \[\left \{ \sum_{r \in S_p(q)} w_q(r):p \in S\setminus \{ q\} \right \}\] can be computed in linear time.
 This implies that the set 
 \[\left \{\left \{ \sum_{r \in S_p(q)} w_q(r): p \in S\setminus \{q\} \right \} : q \in S \right \}\] can be computed in $O(n^2)$ time.
 Therefore,
 \[\left \{\nabla_p : p \in C\right \} \] can be computed in $O(n^2)$ time; the result follows.
 
\subsection*{Proof of Theorem~\ref{thm:add}}

Let $P:=S \cup C$.
For every point $p \in P$, compute the clockwise order of $P \setminus \{p\}$ around $p$;
by Lemma~\ref{lem:SortPoints}, this can be done in $O(n^2)$ time. To every point $p \in P$ assign
a weight of $w(p)=1$ if $p$ is in $S$, and a weight of $w(p)=0$ if $p$ is in $C$.
Use Lemma~\ref{lem:ComputingLemma} to compute the set of values
\[\left \{\sum_{r \in P_p(q)} w(r):p,q \in P \right \}\] in $O(n^2)$ time.
Note that for every pair of points $p,q \in P$ we have that
\[\lambda_S (p,q)=\sum_{r \in P_p(q)} w(r).\] Therefore, 
$f_S(S,S)$, $\{f_S(\{p\},S \cup \{p\}):p \in S\}$ and $\{f_S(S \cup\{ p \},\{p\}):p \in S\}$
can be computed in $O(n^2)$ time.

Note that by Lemma~\ref{lem:FormulaCrossingNumber}, for every $p \in C$ we have that 
 \[\crs(S \cup \{ p\})=f_{S\cup \{ p\} }(S\cup \{ p\},S \cup \{ p\})-\frac{n(n+1)(n-1)(n-2)}{8}.\]
 Thus, it is enough to compute $\{f_{S\cup \{ p\} }(S\cup \{ p\},S\cup \{ p\}): p \in C\}$ in $O(n^2)$ time.
For every $p \in C$, let \[\nabla_p:=f_{S\cup \{ p\}}(S\cup \{ p\},S\cup \{ p\})-f_S(S,S)-f_S(\{p\},S)-f_S(S,\{ p\}).\]
 To compute $\{f_{S\cup \{ p\} }(S\cup \{ p\},S\cup \{ p\}): p \in C\}$ in $O(n^2)$ time, we compute
 $\{\nabla_p : p \in S\}$ in $O(n^2)$ time.
Note that 
\begin{IEEEeqnarray}{rCl}
 \nabla_p & = & f_{S\cup \{ p\}}(S\cup \{ p\},S\cup \{ p\})-f_S(S,S)-f_S(\{p\},S)-f_S(S,\{ p\}) \nonumber \\
          & = &  f_{S\cup \{ p\}}(S\cup \{ p\},S\cup \{ p\})-f_S(S,S)-f_{S \cup \{p\}}(\{p\},S)-f_{S \cup \{p \}}(S,\{ p\}) \nonumber \\
          & = & f_{S\cup \{ p\}}(S,S)-f_S(S,S) \nonumber \\
          & = & \sum_{q\in S} \sum_{r \in S}\left ( \binom{\lambda_{S \cup \{p\}}(q,r)}{2}-\binom{\lambda_{S}(q,r)}{2}\right ). \nonumber
\end{IEEEeqnarray}

Let $r$ be a point in $S$. Note that if
 $p$ is to the right of the directed line from $q$ to $r$ then \[\binom{\lambda_{S\cup \{ p\}}(q,r)}{2}-\binom{\lambda_S(q,r)}{2}=0;\] 
 and if $p$ is to the left of the directed line from $q$ to $r$ then
 \[\binom{\lambda_{S\cup \{ p\}}(q,r)}{2}-\binom{\lambda_S(q,r)}{2}=\lambda_{S}(q,r).\]  Moreover,
 $p$ is to the left of the directed line from $q$ to $r$ if and only if 
 $r$ is to the left of the directed line from $p$ to $q$.
 
 For every point $q \in S$ do the following.
 To every point $r \in S$ 
 assign the weight $w_q(r):=\lambda_{S}(q,r)$. Thus,
 \[ \sum_{r \in S\cup \{ p\}} \left( \binom{\lambda_{S\cup \{ p\}}(q,r)}{2}-\binom{\lambda_S(q,r)}{2} \right ) = \sum_{r \in S_p(q)} w_q(r).\]
  By Lemma~\ref{lem:ComputingLemma}, for a fixed $q$, the set of values
 \[\left \{ \sum_{r \in S_p(q)} w_q(r): p \in C \right \}\] can be computed in linear time.
 This implies that the set 
 \[\left \{\left \{ \sum_{r \in S_p(q)} w_q(r): p \in C \right \} : q \in S \right \}\] can be computed in $O(n^2)$ time.
 Therefore,
 \[\left \{\nabla_p : p \in C\right \} \] can be computed in $O(n^2)$ time; the result follows.

\subsection*{Proof of Theorem~\ref{thm:move}}

 For this it is enough to apply Theorem~\ref{thm:add}
 with $S\setminus\{ p\}$ as the starting set of points, 
 and $C$ as the set of possible new points. 
%

\bibliographystyle{alpha} \bibliography{crossingbib}

\end{document}